\documentclass[english,11pt]{iopart}
\usepackage[T1]{fontenc}
\usepackage[utf8]{inputenc}
\usepackage{geometry}
\usepackage{color}
\usepackage{babel}
\usepackage{enumitem}
\usepackage{amsthm}
\usepackage[unicode=true,
 bookmarks=false,
 breaklinks=true,pdfborder={0 0 0},pdfborderstyle={},backref=false,colorlinks=true]
 {hyperref}
\hypersetup{pdftitle={On the classication of rational K-matrices},
 pdfauthor={Tamás Gombor},
 linkcolor=blue,citecolor=blue}

\makeatletter
\usepackage{iopams}
\usepackage{setstack}
\theoremstyle{plain}
\newtheorem{thm}{\protect\theoremname}
\theoremstyle{definition}
\newtheorem{defn}[thm]{\protect\definitionname}
\theoremstyle{plain}
\newtheorem{prop}[thm]{\protect\propositionname}
\ifx\proof\undefined
\newenvironment{proof}[1][\protect\proofname]{\par
	\normalfont\topsep6\p@\@plus6\p@\relax
	\trivlist
	\itemindent\parindent
	\item[\hskip\labelsep\scshape #1]\ignorespaces
}{%
	\endtrivlist\@endpefalse
}
\providecommand{\proofname}{Proof}
\fi
\theoremstyle{plain}
\newtheorem{cor}[thm]{\protect\corollaryname}
\theoremstyle{remark}
\newtheorem{rem}[thm]{\protect\remarkname}
\theoremstyle{plain}
\newtheorem{lem}[thm]{\protect\lemmaname}

\newcommand{\eqref}[1]{(\ref{#1})}

\makeatother

\providecommand{\corollaryname}{Corollary}
\providecommand{\definitionname}{Definition}
\providecommand{\lemmaname}{Lemma}
\providecommand{\propositionname}{Proposition}
\providecommand{\remarkname}{Remark}
\providecommand{\theoremname}{Theorem}

\begin{document}
\title{On the classification of rational K-matrices}
\author{Tamás Gombor}
\address{Lendület Holographic QFT Group, Wigner Research Centre for Physics,
Konkoly-Thege Miklós u. 29-33, 1121 Budapest , Hungary }
\ead{gombor.tamas@wigner.mta.hu}
\begin{abstract}
This paper presents a derivation of the possible residual symmetries
of rational K-matrices which are invertible in the ''classical limit''
(the spectral parameter goes to infinity). This derivation uses only
the boundary Yang-Baxter equation and the asymptotic expansions of
the R-matrices. The result proves the previous assumption of the literature:
if the original and the residual symmetry algebras are $\mathfrak{g}$
and $\mathfrak{h}$ then there exists a Lie-algebra involution of
$\mathfrak{g}$ for which the invariant sub-algebra is $\mathfrak{h}$.
In addition, we study some K-matrices which are not invertible in
the ''classical limit''. It is shown that their symmetry algebra is
not reductive but a semi-direct sum of reductive and solvable Lie-algebras.

\global\let\newpagegood\newpage
\global\let\newpage\relax
\end{abstract}
\noindent{\it Keywords\/}: {rational K-matrix, boundary Yang-Baxter equation, classification}
\maketitle

\section{Introduction}

\global\let\newpage\newpagegood

Integrable physical systems with boundaries can be defined by their
R- and K-matrices \cite{Sklyanin:1988yz}. The R-matrices satisfy
the famous Yang-Baxter equation and the K-matrix is the solution of
the boundary Yang-Baxter equation (bYBE) for a given R-matrix. In
this paper we give a derivation of the possible residual symmetries
of the rational K-matrices. Some explicit forms of rational K-matrices
(where there is no boundary degrees of freedom) were calculated for
the defining representations of the matrix Lie-algebras \cite{MacKay:2001bh}.
Motivated by a classical argument, the derivation uses the following
assumption. If the boundary breaks the bulk symmetry $G$ to $H$
then $G/H$ has to be a symmetric space, which means that there exists
a Lie-group involution for which the subgroup $H$ is invariant. There
are other direct calculations of K-matrices in the defining representation
in \cite{Arnaudon:2003gj,Arnaudon:2004sd}. 

It is well known that a special class of the ratioanal R-matrices
are classified with the representations of Yangian algebras \cite{Drinfeld:1985rx,Loebbert:2016cdm}.
The introduction of boundaries into the theory of Yangians leads to
a whole new class of the so-called reflection or twisted Yangian algebras
\cite{10.1007/BFb0101183,2002RvMaP..14..317M,Guay2016}. The twisted
Yangians are in exact correspondence with the symmetric spaces $G/H$
which are given by a proper involution. The involution can be used
to construct a co-ideal sub-algebra (which is isomorphic to the twisted
Yangian) of the original Yangian. The representations of this co-ideal
sub-algebra can be further used to calculate K-matrices \cite{Delius:2001he,MacKay:2002at}.

In this paper we derive the assumption above - i.e. we prove for K-matrices
(which are invertible in the ''classical limit'') that the residual
symmetry algebra has to be the invariant sub-algebra of a Lie algebra
involution. This derivation use only the boundary Yang-Baxter equation
and the asymptotic expansion of the R- and K-matrices. We give an
example for a K-matrix with non-invertible ''classical limit'' and
we describe its symmetry which is not a reductive Lie-algebra but
a semi-direct sum of a solvable and a reductive algebra.

The paper is organized as follows. In section 2, we introduce the
basic notations. Section 3 contains the main theorem and proofs which
lead to the classification of possible K-matrices with invertible
''classical limit'' when the boundary vector space is one dimensional.
In this section, we also give an example for K-matrices with non-invertible
''classical limit'' and describe its symmetry. In section 4, we extend
the main theorem of section 3 to K-matrices with general finite dimensional
boundary vector spaces.

\section{Notations}

Let $\mathfrak{g}$ be a complex simple Lie algebra with basis $\left\{ X_{A}\right\} $
where $A=1,\dots,\mathrm{dim(\mathfrak{g})}$ and $\left[X_{A},X_{B}\right]=f_{AB}^{C}X_{C}$.
If $\rho^{(i)}:\mathfrak{g}\to\mathrm{End}(\mathbb{C}^{d_{i}})$ is
a faithful representation of $\mathfrak{g}$ then let $V^{(i)}:=\rho^{(i)}(\mathfrak{g})\subset\mathrm{End}(\mathbb{C}^{d_{i}})$.
There is a non-degenerate invariant symmetric bilinear form (metric)
$\left\langle \cdot,\cdot\right\rangle _{i}:\mathrm{End}(\mathbb{C}^{d_{i}})\otimes\mathrm{End}(\mathbb{C}^{d_{i}})\to\mathbb{C}$
for which $\left\langle Y,Z\right\rangle _{i}=\mathrm{Tr}\left(YZ\right)$.
This metric can be used to define the orthogonal complement of $V^{(i)}$
in $\mathrm{End}(\mathbb{C}^{d_{i}})$: $\mathrm{End}(\mathbb{C}^{d_{i}})=V^{(i)}\oplus\bar{V}^{(i)}$
for which $\left\langle Y,\bar{Y}\right\rangle _{i}=0$ for all $Y\in V^{(i)}$
and $\bar{Y}\in\bar{V}^{(i)}$. We can choose a basis in $V^{(i)}$
and $\bar{V}^{(i)}$. The $\left\{ Y_{A}^{(i)}=\rho^{(i)}(X_{A})\right\} $
is the obvious choice for $V^{(i)}$ and let $\left\{ \bar{Y}_{\bar{A}}^{(i)}\right\} $
be a basis of $\bar{V}^{(i)}$ where $A=1,\dots,\mathrm{dim(\mathfrak{g})}$
and $\bar{A}=\mathrm{dim(\mathfrak{g})}+1,\dots,d^{2}$. The action
of the commutator and the metric on the basis elemets can be written
as
\begin{eqnarray*}
\left[Y_{A}^{(i)},Y_{B}^{(i)}\right] & =f_{AB}^{C}Y_{C}^{(i)} & \left[Y_{A}^{(i)},\bar{Y}_{\bar{B}}^{(i)}\right]=\bar{f}_{A\bar{B}}^{\bar{C}}\bar{Y}_{\bar{C}}^{(i)}\\
\left[\bar{Y}_{\bar{A}}^{(i)},\bar{Y}_{\bar{B}}^{(i)}\right] & =\bar{f}_{\bar{A}\bar{B}}^{C}Y_{C}^{(i)}+\bar{f}_{\bar{A}\bar{B}}^{\bar{C}}\bar{Y}_{\bar{C}}^{(i)}
\end{eqnarray*}
and
\[
C_{AB}^{(i)}=\left\langle Y_{A}^{(i)},Y_{B}^{(i)}\right\rangle _{i}\qquad C_{A\bar{B}}^{(i)}=\left\langle Y_{A}^{(i)},\bar{Y}_{\bar{B}}^{(i)}\right\rangle _{i}=0\qquad C_{\bar{A}\bar{B}}^{(i)}=\left\langle \bar{Y}_{\bar{A}}^{(i)},\bar{Y}_{\bar{B}}^{(i)}\right\rangle _{i}
\]
The metric $C_{AB}^{(i)}$ is proportional to the Killing form $B_{AB}=\mathrm{Tr}\left(\mathrm{ad}_{X_{A}}\circ\mathrm{ad}_{X_{B}}\right)$
i.e. $C_{AB}^{(i)}=c^{(i)}B_{AB}$. Let $B^{AB}$ be the inverse of
$B_{AB}$ i.e. $B^{AB}B_{BC}=\delta_{C}^{A}$. If $\rho^{(1)}$ and
$\rho^{(2)}$ are some representations of $\mathfrak{g}$ then we
can define a useful matrix $C^{(12)}=B^{AB}\rho^{(1)}(X_{A})\otimes\rho^{(2)}(X_{B})$
which is invariant under the action of $\mathfrak{g}$:
\[
\left[C^{(12)},\rho^{(1)}(X)\otimes1+1\otimes\rho^{(2)}(X)\right]=0.
\]

Let $\mathfrak{h}$ be a subalgebra of $\mathfrak{g}$ then ($\mathfrak{g}$,$\mathfrak{h}$)
is a symmetric pair if there exists a Lie-algebra involution $\alpha$
for which $\alpha(X)=X$ for all $X\in\mathfrak{h}$. If $\mathfrak{f}$
is the orthogonal complement of $\mathfrak{h}$ in $\mathfrak{g}$
then $\mathfrak{g}=\mathfrak{h}\oplus\mathfrak{f}$ is a $\mathbb{Z}_{2}$
graded decomposition i.e.
\[
\left[\mathfrak{h},\mathfrak{h}\right]\subseteq\mathfrak{h}\qquad\left[\mathfrak{h},\mathfrak{f}\right]\subseteq\mathfrak{f}\qquad\left[\mathfrak{f},\mathfrak{f}\right]\subseteq\mathfrak{h}
\]
We can choose a basis in $\mathfrak{h}$ and $\mathfrak{f}$: $\left\{ X_{a}\right\} $
and $\left\{ X_{\alpha}\right\} $ respectively where $a=1,\dots,\mathrm{dim(\mathfrak{h})}$
and $\alpha=\mathrm{dim}(\mathfrak{h})+1,\dots,\mathrm{dim(\mathfrak{g})}$.
Using this basis, one can define $C^{(\mathfrak{h},12)}=B^{ab}\rho^{(1)}(X_{a})\otimes\rho^{(2)}(X_{b})$
and $C^{(\mathfrak{f},12)}=B^{\alpha\beta}\rho^{(1)}(X_{\alpha})\otimes\rho^{(2)}(X_{\beta})$
for which $C^{(12)}=C^{(\mathfrak{h},12)}+C^{(\mathfrak{f},12)}$.
The quadratic Casimir of $\mathfrak{h}$ is $c^{(\mathfrak{h},1)}=B^{ab}\rho^{(1)}\left(X_{a}\right)\rho^{(1)}\left(X_{b}\right)$.

Let $\rho^{(B)}:\mathfrak{h}\to\mathrm{End}(\mathbb{C}^{d_{B}})$
be a representation of subalgebra $\mathfrak{h}$ then $C^{(\mathfrak{h},1B)}=B^{ab}\rho^{(1)}(X_{a})\otimes\rho^{(B)}(X_{b})$.

Let $W^{(i)}:=\rho^{(i)}(\mathfrak{h})\subset\mathrm{End}(\mathbb{C}^{d_{i}})$.
Obviously $W^{(i)}\subseteq V^{(i)}$. The metric can be used to define
the orthogonal complement of $W^{(i)}$ in $\mathrm{End}(\mathbb{C}^{d_{i}})$:
$\mathrm{End}(\mathbb{C}^{d_{i}})=W^{(i)}\oplus\bar{W}^{(i)}$. We
can choose bases $\left\{ Y_{a}^{(i)}=\rho(X_{a})\right\} $ and $\left\{ \bar{Y}_{\bar{a}}^{(i)}\right\} $
in $W^{(i)}$ and $\bar{W}^{(i)}$ respectively where $a=1,\dots,\mathrm{dim(\mathfrak{h})}$
and $\bar{a}=\mathrm{dim(\mathfrak{h})}+1,\dots,d^{2}$.

\begin{defn}
The quasi-classical R-matrix in the representation $\rho^{(i)}\otimes\rho^{(j)}$
is a $R^{(ij)}(u):\mathbb{C}^{d_{i}}\otimes\mathbb{C}^{d_{j}}\to\mathbb{C}^{d_{i}}\otimes\mathbb{C}^{d_{j}}$
spectral parameter ($u\in\mathbb{C}$) dependent linear map which
satisfy the Yang-Baxter equation
\[
R_{12}^{(12)}(u)R_{13}^{(13)}(u+v)R_{23}^{(23)}(v)=R_{23}^{(23)}(v)R_{13}^{(13)}(u+v)R_{12}^{(12)}(u)
\]
and its asymptotic expansion is:
\begin{equation}
R^{(ij)}(u)=1+\frac{1}{u}C^{(ij)}+\mathcal{O}(u^{-2}),\label{eq:assymR}
\end{equation}
where $\rho^{(i)}:\mathfrak{g}\to\mathrm{End}(\mathbb{C}^{d_{i}})$
are some representations for $i=1,2,3$ and 
\begin{eqnarray*}
R_{12}^{(12)}(u) & =R^{(12)}(u)\otimes1 & R_{23}^{(23)}(u)=1\otimes R^{(23)}(u)\\
R_{13}^{(13)}(u) & =P_{23}\left(R^{(13)}(u)\otimes1\right)P_{32}
\end{eqnarray*}
\end{defn}

\begin{defn}
Let $R^{(ij)}(u)\in\mathrm{End}\left(\mathbb{C}^{d_{i}}\otimes\mathbb{C}^{d_{j}}\right)$
be a quasi classical R-matrix in the representation $\rho^{(i)}\otimes\rho^{(j)}$
and $\mathfrak{h}$ is a sub-algebra of $\mathfrak{g}$ where $i,j=1,2$.
The map $K(u):\mathbb{C}^{d_{2}}\otimes\mathbb{C}^{d_{B}}\to\mathbb{C}^{d_{1}}\otimes\mathbb{C}^{d_{B}}$
is a $(\mathfrak{g},\mathfrak{h})$ symmetric K-matrix in the representation
$\left(\rho^{(1)},\rho^{(2)}\right)$ if the following two conditions
are satisfied.
\end{defn}
\begin{itemize}
\item There exists a representation $\rho^{(B)}:\mathfrak{g}\to\mathrm{End}(\mathbb{C}^{d_{B}})$
for which 
\begin{equation}
\rho^{(1)}(X)K(u)-K(u)\rho^{(2)}(X)+\left[\rho^{(B)}(X),K(u)\right]=0\label{eq:symK}
\end{equation}
for all $X\in\mathfrak{h}\subset\mathfrak{g}$.
\item It satisfies the \emph{boundary Yang-Baxter equation} (bYBE):
\end{itemize}
\[
R_{12}^{(11)}(u-v)K_{13}(u)R_{21}^{(12)}(u+v)K_{23}(v)=K_{23}(v)R_{12}^{(12)}(u+v)K_{13}(u)R_{21}^{(22)}(u-v).
\]

The $\mathbb{C}^{d_{B}}$ is the boundary vector space.
\begin{defn}
Let $K(u)$ be a $(\mathfrak{g},\mathfrak{h})$ symmetric K-matrix
with the following asymptotic expansion
\begin{equation}
K(u)=\kappa+\mathcal{O}(u^{-1}).\label{eq:assymK}
\end{equation}
The K-matrix $K(u)$ is called quasi-classical if $\kappa$ is invertible.
\end{defn}
In the following we deal with quasi-classical K-matrices therefore
$d_{1}=d_{2}=d$.

\section{K-matrices with 1 dimensional boundary space}

In this section we investigate K-matrices with one dimensional boundary
space i.e. $d_{B}=1$.

\subsection{Classification of possible quasi-classical K-matrices}
\begin{prop}
\label{lem:Adk}If there exists a quasi classical K-matrix in the
representation $\left(\rho^{(1)},\rho^{(2)}\right)$ then the map
$\mathrm{Ad}_{\kappa}:\mathrm{End}(\mathbb{C}^{d})\to\mathrm{End}(\mathbb{C}^{d})$
is a bijection between $\rho^{(1)}(\mathfrak{g})$ and $\rho^{(2)}(\mathfrak{g})$
i.e. $\mathrm{Ad}_{\kappa}(\rho^{(2)}(X_{A}))=M_{A}^{B}\rho^{(1)}(X_{B})$
where $\mathbf{M}$ is invertible, and $\mathbf{M}^{2}=1$ ($\left(\mathbf{M}\right)_{A}^{B}=M_{A}^{B}$). 
\end{prop}
\begin{proof}
At first, we use the bYBE
\begin{eqnarray}
R_{12}^{(11)}\left(\frac{u-v}{x}\right)K_{1}\left(\frac{u}{x}\right)R_{21}^{(12)}\left(\frac{u+v}{x}\right)K_{2}\left(\frac{v}{x}\right) & =\nonumber \\
K_{2}\left(\frac{v}{x}\right)R_{12}^{(12)}\left(\frac{u+v}{x}\right)K_{1}\left(\frac{u}{x}\right)R_{21}^{(22)}\left(\frac{u-v}{x}\right)\label{eq:bYBE1}
\end{eqnarray}
in the $x\to0$ limit. The first non trivial term is:
\begin{equation}
\frac{x}{u-v}\left(C^{(11)}\kappa_{1}\kappa_{2}-\kappa_{1}\kappa_{2}C^{(22)}\right)+\frac{x}{u+v}\left(\kappa_{1}C^{(21)}\kappa_{2}-\kappa_{2}C^{(12)}\kappa_{1}\right)=\mathcal{O}(x^{2}),\label{eq:classlimit}
\end{equation}
where we used the expansions \eqref{eq:assymR},\eqref{eq:assymK}.
This is the \emph{classical boundary Yang-Baxter equation} (cbYBE).
The above equation is equivalent to two constraints on $\kappa$:
\begin{eqnarray}
C^{(11)}\kappa_{1}\kappa_{2} & = & \kappa_{1}\kappa_{2}C^{(22)}\label{eq:Ckk}\\
\kappa_{1}C^{(21)}\kappa_{2} & = & \kappa_{2}C^{(12)}\kappa_{1}\label{eq:kCk}
\end{eqnarray}
Equation \eqref{eq:Ckk} can be written in the following form: 
\begin{equation}
B^{AB}Y_{A}^{(1)}\otimes Y_{B}^{(1)}=B^{AB}\kappa Y_{A}^{(2)}\kappa^{-1}\otimes\kappa Y_{B}^{(2)}\kappa^{-1}=B^{AB}\mathrm{Ad}_{\kappa}(Y_{A}^{(2)})\otimes\mathrm{Ad}_{\kappa}(Y_{B}^{(2)})\label{eq:lemma11}
\end{equation}
From this, we can derive that $\mathrm{Ad}_{\kappa}$ is a bijection
between $\rho^{(1)}(\mathfrak{g})$\textit{ }and\textit{ }$\rho^{(2)}(\mathfrak{g})$.
The action of $\mathrm{Ad}_{\kappa}$ on $\rho^{(2)}(\mathfrak{g})$
can be written as follows 
\[
\mathrm{Ad}_{\kappa}(Y_{A}^{(2)})=M_{A}^{B}Y_{B}^{(1)}+N_{A}^{\bar{B}}\bar{Y}_{\bar{B}}^{(1)}.
\]
Appling the operator $1\otimes\left\langle Y_{C}^{(1)},\cdot\right\rangle _{1}$
on \eqref{eq:lemma11} we obtain that
\[
B^{AB}C_{BC}^{(1)}Y_{A}^{(1)}=B^{AB}\left(M_{A}^{D}Y_{D}^{(1)}+N_{A}^{\bar{D}}\bar{Y}_{\bar{D}}^{(1)}\right)M_{B}^{E}C_{EC}^{(1)}.
\]
Since $C_{AB}^{(1)}=c^{(1)}B_{AB}$ 
\[
Y_{C}^{(1)}=B^{AB}\left(M_{A}^{D}Y_{D}^{(1)}+N_{A}^{\bar{D}}\bar{Y}_{\bar{D}}^{(1)}\right)M_{B}^{E}B_{EC}.
\]
From this, we obtain two constrains for $\mathbf{M}$ and $\mathbf{N}$:
\begin{eqnarray*}
\mathbf{B}\mathbf{M}^{T}\mathbf{B}^{-1}\mathbf{M} & = & 1,\\
\mathbf{B}\mathbf{M}^{T}\mathbf{B}^{-1}\mathbf{N} & = & 0.
\end{eqnarray*}
It follows from the first that $\mathbf{M}$ is invertible and $\mathbf{M}^{-1}=\mathbf{B}\mathbf{M}^{T}\mathbf{\mathbf{B}}^{-1}$.
Using this in the second equation, we obtain that $\mathbf{N}=0$.
Therefore 
\begin{equation}
\mathrm{Ad}_{\kappa}(Y_{A}^{(2)})=M_{A}^{B}Y_{B}^{(1)}\label{eq:Adk}
\end{equation}
which implies that $\mathrm{Ad}_{\kappa}$ is a bijection between
$\rho^{(1)}(\mathfrak{g})$ and $\rho^{(2)}(\mathfrak{g})$.

The equation \eqref{eq:kCk} can be written in the following form:
\[
B^{AB}\mathrm{Ad}_{\kappa}(Y_{A}^{(2)})\otimes Y_{B}^{(1)}=B^{AB}Y_{A}^{(1)}\otimes\mathrm{Ad}_{\kappa}(Y_{B}^{(2)}).
\]
Using \eqref{eq:Adk} and applying $1\otimes\left\langle Y_{C}^{(1)},\cdot\right\rangle _{1}$
\[
M_{C}^{D}Y_{D}^{(1)}=B^{AB}Y_{A}^{(1)}M_{B}^{E}B_{EC}.
\]
Therefore 
\[
\mathbf{M}=\mathbf{B}\mathbf{M}^{T}\mathbf{B}^{-1},
\]
but we have seen above that $\mathbf{M}^{-1}=\mathbf{B}\mathbf{M}^{T}\mathbf{B}^{-1}$
therefore $\mathbf{M}^{2}=1$.
\end{proof}
\begin{cor}
\label{cor:inv}If there exists a quasi classical K-matrix in the
representation $\left(\rho^{(1)},\rho^{(2)}\right)$ then there exists
a Lie algebra involution $\alpha:\mathfrak{g}\to\mathfrak{g}$, $\alpha^{2}=\mathrm{id}_{\mathfrak{g}}$
for which \textup{$\mathrm{Ad}_{\kappa}(\rho^{(2)}(X))=\rho^{(1)}(\alpha(X))$}.
\end{cor}
\begin{proof}
Since $\mathrm{Ad}_{\kappa}(\rho^{(2)}(X))\in\rho^{(1)}(\mathfrak{g})$
for all $X\in\mathfrak{g}$ and $\rho^{(1)}$ is faithful, the map
$\alpha=\left(\rho^{(1)}\right)^{-1}\circ\mathrm{Ad}_{\kappa}\circ\rho^{(2)}:\mathfrak{g}\to\mathfrak{g}$
exists and it is a Lie algebra automorphism. From \eqref{eq:Adk}
we obtain that $\alpha(X_{A})=M_{A}^{B}X_{B}$. We also saw that $\mathbf{M}^{2}=1$,
therefore $\alpha^{2}=\mathrm{id}_{\mathfrak{g}}$ i.e. $\alpha$
is a Lie algebra involution.

From the definition of $\alpha$ we obtain that $\alpha(X)=\left(\rho^{(1)}\right)^{-1}\left(\mathrm{Ad}_{\kappa}\left(\rho^{(2)}(X)\right)\right)$
therefore 
\begin{equation}
\mathrm{Ad}_{\kappa}(\rho^{(2)}(X))=\rho^{(1)}(\alpha(X))\label{eq:Adkalp}
\end{equation}
for all $X\in\mathfrak{g}$.
\end{proof}
\begin{rem}
The equation \eqref{eq:Adkalp} can be written as
\[
\rho^{(2)}=\mathrm{Ad}_{\kappa^{-1}}\circ\rho^{(1)}\circ\alpha
\]
Therefore if we choose an arbitrary representation $\rho^{(1)}$ and
a Lie algebra involution $\alpha$ then the equation above fixes $\rho^{(2)}$.
Let us choose $\rho^{(1)}=\rho$ then there are two possibilities
for $\rho^{(2)}$.
\end{rem}
\begin{enumerate}[label=(\roman*)]
\item  There exists $V\in\mathrm{Aut}(\mathbb{C}^{d})$ such that $\rho\circ\alpha=\mathrm{Ad}_{V^{-1}}\circ\rho$.
This $V$ exists for all inner and some outer involutions. Using $V$,
the equation \eqref{eq:Adkalp} reads as
\[
\rho^{(2)}=\mathrm{Ad}_{\left(V\kappa\right)^{-1}}\circ\rho
\]
which means $\rho$ and $\rho^{(2)}$ are equivalent representations,
it is therefore advisable to choose a basis where $\rho^{(2)}=\rho$.
In this basis
\[
\rho(\alpha(X))=\mathrm{Ad}_{\kappa}(\rho(X))=\kappa\rho(X)\kappa^{-1},
\]
Therefore these K-matrices belong to the usual untwisted boundary
Yang-Baxter equation:
\[
R_{12}(u-v)K_{1}(u)R_{21}(u+v)K_{2}(v)=K_{2}(v)R_{12}(u+v)K_{1}(u)R_{21}(u-v),
\]
where $R(u)$ is the R-matrix in the $\rho\otimes\rho$ representation.
\item There is no $V\in\mathrm{Aut}(\mathbb{C}^{d})$ such that $\rho\circ\alpha=\mathrm{Ad}_{V^{-1}}\circ\rho$
which is true for some outer involutions. These belong to the $\mathbb{Z}_{2}$
automorphisms of the Dynkin-diagrams of the Lie-algebras, e.g. in
the case of the $A_{n}$ series, these connect the representations
to their contra-gradient representations, therefore $\rho^{(2)}=\rho_{cg}$
i.e. 
\[
\rho(\alpha(X))=\mathrm{Ad}_{\kappa}(\rho_{cg}(X))=-\kappa\rho(X)^{T}\kappa^{-1},
\]
for all $X\in\mathfrak{g}$. Therefore these K-matrices belong to
the twisted boundary Yang-Baxter equation:
\[
R_{12}(u-v)K_{1}(u)\bar{R}_{21}(u+v)K_{2}(v)=K_{2}(v)\bar{R}_{12}(u+v)K_{1}(u)R_{21}(u-v),
\]
where $\bar{R}(u)$ is the crossed R-matrix of $R(u)$:
\[
\bar{R}(u)=R(\Gamma-u)^{T_{1}},
\]
where $\Gamma$ is the crossing parameter.
\end{enumerate}
\begin{thm}
Let $K(u)$ be a quasi classical $(\mathfrak{g},\mathfrak{h})$ symmetric
K-matrix in the representation $\left(\rho^{(1)},\rho^{(2)}\right)$
then $(\mathfrak{g},\mathfrak{h})$ is a symmetric pair.
\end{thm}
\begin{proof}
From the previous corollary there exists a Lie algebra involution
$\alpha$, for which\textit{ }$\mathrm{Ad}_{\kappa}(\rho^{(2)}(X))=\rho^{(1)}(\alpha(X))$\textit{.}
This involution can be used for a $\mathbb{Z}_{2}$ graded decomposition:
$\mathfrak{g}=\mathfrak{h}_{0}\oplus\mathfrak{f}$ where $\alpha(X^{(+)})=+X^{(+)}$
and $\alpha(X^{(-)})=-X^{(-)}$ for all $X^{(+)}\in\mathfrak{h}_{0}$
and $X^{(-)}\in\mathfrak{f}$. Therefore $(\mathfrak{g},\mathfrak{h}_{0})$
is a symmetric pair.

Let us assume that $X\in\mathfrak{h}$. This implies that $K(u)\rho^{(2)}(X)=\rho^{(1)}(X)K(u)$
(see equation \eqref{eq:symK}) which reads as $\kappa\rho^{(2)}(X)=\rho^{(1)}(X)\kappa$
in the asymtotic limit. Therefore $\mathrm{Ad}_{\kappa}(\rho^{(2)}(X))=\rho^{(1)}(X)$
which implies that $\alpha(X)=X$ i.e $X\in\mathfrak{h}_{0}$ therefore
$\mathfrak{h}\subseteq\mathfrak{h}_{0}$.

Now, we take the bYBE in the $v\to\infty$ limit:
\begin{eqnarray*}
\frac{1}{v}B^{AB}\left(K(u)\rho^{(2)}(X_{A})-\rho^{(1)}(X_{A})K(u)\right)\otimes\rho^{(1)}(X_{B})\kappa+\mathcal{O}(v^{-2}) & =\\
\frac{1}{v}B^{AB}\left(\rho^{(1)}(X_{A})K(u)-K(u)\rho^{(2)}(X_{A})\right)\otimes\kappa\rho^{(2)}(X_{B})+\mathcal{O}(v^{-2})
\end{eqnarray*}
which implies that
\begin{eqnarray*}
B^{AB}\left(K(u)\rho^{(2)}(X)-\rho^{(1)}(X)K(u)\right)\otimes\left(\rho^{(1)}(X_{B})+\kappa\rho^{(2)}(X_{B})\kappa^{-1}\right) & =\\
=B^{AB}\left(K(u)\rho^{(2)}(X)-\rho^{(1)}(X)K(u)\right)\otimes\rho^{(1)}(X_{B}+\alpha(X_{B})) & = & 0
\end{eqnarray*}
Therefore
\[
B^{ab}\left(K(u)\rho^{(2)}(X_{a})-\rho^{(1)}(X_{a})K(u)\right)\otimes\rho^{(1)}(X_{b})=0
\]
which is equivalent to
\[
K(u)\rho^{(2)}(X)=\rho^{(1)}(X)K(u)
\]
for all $X\in\mathfrak{h}_{0}$ which implies that $\mathfrak{h}_{0}\subseteq\mathfrak{h}$.
We have seen previously that $\mathfrak{h}\subseteq\mathfrak{h}_{0}$
therefore $\mathfrak{h}_{0}=\mathfrak{h}$ i.e. $\left(\mathfrak{g},\mathfrak{h}\right)$
is a symmetric pair.
\end{proof}
\begin{thm}
\label{lem:univ}The $(\mathfrak{g},\mathfrak{h})$ symmetric K-matrix
in the representation $\left(\rho^{(1)},\rho^{(2)}\right)$ is unique
up to a multiplicative scalar function if $\mathfrak{h}$ is semi-simple,
\textup{$\rho^{(1)}$} is irreducible and the matrix $\kappa$ is
fixed. If $\mathfrak{h}$ is not semi-simple but reductive then the
K-matrix may have a free parameter.
\end{thm}
\begin{rem}
The proof of this theorem can be found in \cite{Bittleston:2019gkq}
for a special case when $\mathfrak{g}=\mathfrak{sl}(n)$ and $\rho^{(1)}$
and $\rho^{(2)}$ are the defining representations. The authors of
that paper uses the Sklyanin determinant but we do not use it in the
following modified proof.
\end{rem}
\begin{proof}
Let the asymptotic expansion of $K(u)$ be
\[
K(u)=\kappa+\frac{1}{u}k^{(1)}+\dots+\frac{1}{u^{r}}k^{(r)}+\dots
\]
where $\kappa$ is a fixed. The bYBE is invariant under the scalar
multiplication: $K(u)\to c(u)K(u)$. We can fix this freedom by
\begin{equation}
\mathrm{Tr}\left(K(u)\kappa^{-1}\right)=d.\label{eq:norm-1}
\end{equation}

Let us assume that we have two K-matrices satisfying the bYBE with
the above normalization, and these agree up to order $r-1$. The difference
of the two K-matrices are
\[
\tilde{K}(u)-K(u)=\frac{1}{u^{r}}\delta k^{(r)}+\dots
\]
From the normalization \eqref{eq:norm-1} we obtain that
\begin{equation}
\mathrm{Tr}\left(\delta k^{(r)}\kappa^{-1}\right)=0.\label{eq:norm}
\end{equation}

Substituting $K$ and $\tilde{K}$ into bYBE \eqref{eq:bYBE1} and
subtracting them from each other, the leading non-trivial terms at
lowest order in $x$ are the following:
\begin{eqnarray*}
 & \frac{x^{r+1}}{u^{r}(u-v)} & \left(C^{(11)}\delta k_{1}^{(r)}\kappa_{2}-\delta k_{1}^{(r)}\kappa_{2}C^{(22)}\right)+\\
+ & \frac{x^{r+1}}{v^{r}(u-v)} & \left(C^{(11)}\kappa_{1}\delta k_{2}^{(r)}-\kappa_{1}\delta k_{2}^{(r)}C^{(22)}\right)+\\
+ & \frac{x^{r+1}}{u^{r}(u+v)} & \left(\delta k_{1}^{(r)}C^{(21)}\kappa_{2}-\kappa_{2}C^{(12)}\delta k_{1}^{(r)}\right)+\\
+ & \frac{x^{r+1}}{v^{r}(u+v)} & \left(\kappa_{1}C^{(21)}\delta k_{2}^{(r)}-\delta k_{2}^{(r)}C^{(12)}\kappa_{1}\right)=\mathcal{O}(x^{r+2})
\end{eqnarray*}

The spectral parameter dependent functions 
\[
\frac{1}{u^{r}(u-v)}\qquad\frac{1}{v^{r}(u-v)}\qquad\frac{1}{u^{r}(u+v)}\qquad\frac{1}{v^{r}(u+v)}
\]
are linearly independent for $r>1$ and linearly dependent for $r=1$:
\[
\frac{1}{u(u-v)}-\frac{1}{v(u-v)}+\frac{1}{u(u+v)}+\frac{1}{v(u+v)}=0.
\]
For $r>1$ we have 4 constrains. Let us see the first one:
\[
C^{(11)}\delta k_{1}^{(r)}\kappa_{2}=\delta k_{1}^{(r)}\kappa_{2}C^{(22)}
\]
or in an equivalent form:
\[
C^{(11)}\delta k_{1}^{(r)}=\delta k_{1}^{(r)}\kappa_{2}C^{(22)}\kappa_{2}^{-1}.
\]
Using \eqref{eq:Ckk}, this can be written as
\[
C^{(11)}\delta k_{1}^{(r)}=\delta k_{1}^{(r)}\kappa_{1}^{-1}C^{(11)}\kappa_{1}
\]
i.e.
\[
\left[C^{(11)},\delta k_{1}^{(r)}\kappa_{1}^{-1}\right]=B^{AB}\left[\rho^{(1)}(X_{A}),\delta k^{(r)}\kappa^{-1}\right]\otimes\rho^{(1)}(X_{B})=0,
\]
which implies that
\[
\left[\rho^{(1)}(X),\delta k^{(r)}\kappa^{-1}\right]=0
\]
for all $X\in\mathfrak{g}$. Since $\rho^{(1)}$ is irreducible, $\delta k^{(r)}\kappa^{-1}$
has to be proportional to the identity but from \eqref{eq:norm} we
can see that $\delta k^{(r)}$ has to vanish.

For $r=1$ we have three equations 
\begin{eqnarray}
\left(C^{(11)}\kappa_{1}\delta k_{2}^{(1)}-\kappa_{1}\delta k_{2}^{(1)}C^{(22)}\right)+\left(C^{(11)}\delta k_{1}^{(1)}\kappa_{2}-\delta k_{1}^{(1)}\kappa_{2}C^{(22)}\right) & = & 0\label{eq:r11}\\
\left(C^{(11)}\kappa_{1}\delta k_{2}^{(1)}-\kappa_{1}\delta k_{2}^{(1)}C^{(22)}\right)+\left(\delta k_{1}^{(1)}C^{(21)}\kappa_{2}-\kappa_{2}C^{(12)}\delta k_{1}^{(1)}\right) & = & 0\label{eq:r12}\\
\left(C^{(11)}\kappa_{1}\delta k_{2}^{(1)}-\kappa_{1}\delta k_{2}^{(1)}C^{(22)}\right)+\left(\kappa_{1}C^{(21)}\delta k_{2}^{(1)}-\delta k_{2}^{(1)}C^{(12)}\kappa_{1}\right) & = & 0\label{eq:r13}
\end{eqnarray}
We can see that equation \eqref{eq:r12} follows from \eqref{eq:r11}
and \eqref{eq:r13}, therefore we only have to deal with these two.
Let us start with equation \eqref{eq:r11}. Multiplying by $\kappa_{1}^{-1}\kappa_{2}^{-1}$
from the right:
\[
C^{(11)}\delta k_{1}^{(1)}\kappa_{1}^{-1}+C^{(11)}\delta k_{2}^{(1)}\kappa_{2}^{-1}=\delta k_{1}^{(1)}\kappa_{2}C^{(22)}\kappa_{1}^{-1}\kappa_{2}^{-1}+\kappa_{1}\delta k_{2}^{(1)}C^{(22)}\kappa_{1}^{-1}\kappa_{2}^{-1}
\]
Using \eqref{eq:Ckk} we obtain that
\[
C^{(11)}\delta k_{1}^{(1)}\kappa_{1}^{-1}+C^{(11)}\delta k_{2}^{(1)}\kappa_{2}^{-1}=\delta k_{1}^{(1)}\kappa_{1}^{-1}C^{(11)}+\delta k_{2}^{(1)}\kappa_{2}^{-1}C^{(11)}
\]
i.e.
\begin{eqnarray*}
\left[C^{(11)},\delta k_{1}^{(1)}\kappa_{1}^{-1}+\delta k_{2}^{(1)}\kappa_{2}^{-1}\right]= &  & ,\\
=B^{AB}\left[\rho^{(1)}(X_{A}),\delta k^{(1)}\kappa^{-1}\right]\otimes\rho^{(1)}(X_{B})+B^{AB}\rho^{(1)}(X_{A})\otimes\left[\rho^{(1)}(X_{B}),\delta k^{(1)}\kappa^{-1}\right] & = & 0
\end{eqnarray*}
Using the basis in $\mathrm{End}(\mathbb{C}^{d})$ related to $\rho^{(1)}$:
$\delta k^{(1)}\kappa^{-1}=Z^{A}Y_{A}^{(1)}+Z^{\bar{A}}\bar{Y}_{\bar{A}}^{(1)}$,
we obtain that
\begin{eqnarray*}
B^{AB}\left[Y_{A}^{(1)},Z^{C}Y_{C}^{(1)}+Z^{\bar{C}}\bar{Y}_{\bar{C}}^{(1)}\right]\otimes Y_{B}^{(1)}+B^{AB}Y_{A}^{(1)}\otimes\left[Y_{B}^{(1)},Z^{C}Y_{C}^{(1)}+Z^{\bar{C}}\bar{Y}_{\bar{C}}^{(1)}\right]=\\
=B^{AB}Z^{\bar{C}}f_{A\bar{C}}^{\bar{D}}\bar{Y^{(1)}}_{\bar{D}}\otimes Y_{B}^{(1)}+B^{AB}Z^{\bar{C}}f_{B\bar{C}}^{\bar{D}}Y_{A}^{(1)}\otimes\bar{Y}_{\bar{D}}^{(1)}=0.
\end{eqnarray*}
Applying the operator $1\otimes\left\langle Y_{D}^{(1)},\cdot\right\rangle _{1}$we
obtain that
\[
Z^{\bar{C}}f_{D\bar{C}}^{\bar{D}}Y_{\bar{D}}^{(1)}=\left[Y_{D}^{(1)},Z^{\bar{C}}\bar{Y}_{\bar{C}}^{(1)}\right]=0.
\]
 which implies that
\[
\left[\rho(X),Z^{\bar{C}}\bar{Y}_{\bar{C}}^{(1)}\right]=0
\]
for all $X\in\mathfrak{g}$. Because of $\rho^{(1)}$ is irreducible,
$Z^{\bar{C}}\bar{Y}_{\bar{C}}^{(1)}=c1$ where $c\in\mathbb{C}$,
therefore $\delta k^{(1)}\kappa^{-1}=Z^{A}Y_{A}^{(1)}+c1$. From \eqref{eq:norm},
we can obtain that $c=0$ which implies that $\delta k^{(1)}\kappa^{-1}\in\rho^{(1)}(\mathfrak{g})$.

Let us continue with \eqref{eq:r13}
\[
C^{(11)}\delta k_{2}^{(1)}\kappa_{2}^{-1}+\kappa_{1}C^{(21)}\kappa_{1}^{-1}\delta k_{2}^{(1)}\kappa_{2}^{-1}=\kappa_{1}\delta k_{2}^{(1)}C^{(22)}\kappa_{1}^{-1}\kappa_{2}^{-1}+\delta k_{2}^{(1)}C^{(12)}\kappa_{2}^{-1}
\]
Using \eqref{eq:Ckk} and \eqref{eq:kCk} we can obtain that
\[
C^{(11)}\delta k_{2}^{(1)}\kappa_{2}^{-1}+\kappa_{1}C^{(21)}\kappa_{1}^{-1}\delta k_{2}^{(1)}\kappa_{2}^{-1}=\delta k_{2}^{(1)}\kappa_{2}^{-1}C^{(11)}+\delta k_{2}^{(1)}\kappa_{2}^{-1}\kappa_{1}C^{(21)}\kappa_{1}^{-1}
\]
which implies that
\[
B^{AB}\left(\rho^{(1)}(X_{A})+\kappa\rho^{(2)}(X_{A})\kappa^{-1}\right)\otimes\left[\rho^{(1)}(X_{B}),\delta k^{(1)}\kappa^{-1}\right]=0.
\]
Using \eqref{eq:Adkalp}, we obtain that
\begin{eqnarray*}
B^{AB}\rho^{(1)}(X_{A}+\alpha(X_{A}))\otimes\left[\rho^{(1)}(X_{B}),\delta k^{(1)}\kappa^{-1}\right]=\\
B^{ab}\rho^{(1)}(X_{a})\otimes\left[\rho^{(1)}(X_{b}),\delta k^{(1)}\kappa^{-1}\right]=0
\end{eqnarray*}
therefore 
\[
\left[\rho^{(1)}(X),\delta k^{(1)}\kappa^{-1}\right]=0
\]
for all $X\in\mathfrak{h}$. Since $\delta k^{(1)}\kappa^{-1}\in\rho^{(1)}(\mathfrak{g})$,
$\delta k^{(1)}=0$ if $\mathfrak{h}$ is semi-simple. For reductive
$\mathfrak{h}$, the $\delta k^{(1)}\kappa^{-1}$ is an element of
the center of $\mathfrak{h}$ which is a one dimensional subspace
for the $(\mathfrak{g},\mathfrak{h})$ symmetric pair (see the classification
of symmetric pairs \cite{helgason1979differential}), therefore $K(u)$
may has a free parameter.
\end{proof}

\subsection{Comments on the classical limit}

The reason we use the ''quasi classical'' name is the following \cite{Drinfeld:1985rx}.
The YBE does not fix the norm of the R-matrix and the spectral parameter
(if $R(u)$ a solution then $c(u)R(xu)$ is also a solution for all
$x\in\mathbb{C}$ and any complex function $c(u)$) therefore one
can redefine the spectral parameter as $u\to u/\hbar$ for which
\[
R^{(ij)}(u,\hbar):=R^{(ij)}(u/\hbar)=1+\hbar r^{(ij)}(u)+\mathcal{O}(\hbar^{2}).
\]
where $r(u)$ is the \emph{classical R-matrix}
\[
r^{(ij)}(u)=\frac{1}{u}C^{(ij)}
\]
 which satisfies the \emph{classical Yang-Baxter equation}
\[
\left[r_{12}^{(12)}(u),r_{13}^{(13)}(u+v)\right]+\left[r_{12}^{(12)}(u),r_{23}^{(23)}(v)\right]+\left[r_{13}^{(13)}(u+v),r_{23}^{(23)}(v)\right]=0
\]
which is the order $\hbar^{2}$ term of the YBE. Scaling the spectral
parameter in the K-matrix similarly, we obtain that
\[
K(u,\hbar):=K(u/\hbar)=\kappa+\mathcal{O}(\hbar).
\]
The first non-trivial term of the bYBE in $\hbar$ reads as
\[
r^{(11)}(u-v)\kappa_{1}\kappa_{2}-\kappa_{1}\kappa_{2}r^{(22)}(u-v)+\kappa_{1}r^{(21)}(u+v)\kappa_{2}-\kappa_{2}r^{(12)}(u+v)\kappa_{1}=0
\]
which is the \emph{classical boundary Yang-Baxter equation} (cbYBE)
which was investigated in Proposition \ref{lem:Adk}.

However, we saw that if the residual symmetry is not semi-simple then
the K-matrix may have a free parameter ($K(u)\to K(u,a)$ where $a$
is the free parameter). These free parameter dependent solutions really
exist \cite{Gombor:2017qsy}. The asymptotic expansion of these K-matrices
can be written as
\begin{equation}
K(u,a)=\kappa+\frac{1}{u}\left(k^{(1)}+a\rho^{(1)}(X_{0})\right)\kappa+\mathcal{O}(u^{-2}),\label{eq:k1}
\end{equation}
where $X_{0}$ is the central element of $\mathfrak{h}$ and
\[
\left\langle \rho^{(1)}(X_{0}),k^{(1)}\right\rangle =0\qquad\left\langle \rho^{(1)}(X_{0}),\rho^{(1)}(X_{0})\right\rangle =1
\]
The above parameter can become an $\hbar$ dependent function: $a\to a(\hbar)$.
Using a proper function $a(\hbar)$ the K-matrix can be written as
\begin{equation}
K(u,\hbar,a(\hbar)):=K(u/\hbar,a(\hbar))=\tilde{\kappa}(u)+\mathcal{O}(\hbar)\label{eq:k2}
\end{equation}
which means that the classical limit of the K-matrix can be spectral
parameter dependent. Using this $\hbar$ expansion, the first non-trivial
term of the bYBE can be written as
\begin{eqnarray*}
r_{12}^{(11)}(u-v)\tilde{\kappa}_{1}(u)\tilde{\kappa}_{2}(v)-\tilde{\kappa}_{1}(u)\tilde{\kappa}_{2}(v)r_{12}^{(22)}(u-v)+\\
\tilde{\kappa}_{1}(u)r_{12}^{(21)}(u+v)\tilde{\kappa}_{2}(v)-\tilde{\kappa}_{2}(v)r_{12}^{(12)}(u+v)\tilde{\kappa}_{1}(u)=0
\end{eqnarray*}
which is the classical boundary Yang-Baxter equation for \emph{spectral
parameter dependent $\kappa$-matrix}. In \cite{Gombor:2018ppd},
there was derived some solutions of this equation in the defining
representations of the matrix Lie-algebras. These solutions can be
matched to the parameter dependent solutions of the bYBE.

One can classify the solutions of the general cbYBE
\begin{eqnarray}
\frac{1}{u-v}\left(C_{12}^{(11)}\tilde{\kappa}_{1}(u)\tilde{\kappa}_{2}(v)-\tilde{\kappa}_{1}(u)\tilde{\kappa}_{2}(v)C_{12}^{(22)}\right)+\label{eq:cbYBE}\\
\frac{1}{u+v}\left(\tilde{\kappa}_{1}(u)C_{12}^{(21)}\tilde{\kappa}_{2}(v)-\tilde{\kappa}_{2}(v)C_{12}^{(12)}\tilde{\kappa}_{1}(u)\right)=0
\end{eqnarray}
if we assume that $\tilde{\kappa}(u)$ is invertable in the asymtotic
limit i.e
\[
\tilde{\kappa}(u)=\kappa+\mathcal{O}(u^{-1})
\]
where $\kappa$ is invertable. This $\kappa$ also satisfies the spectral
parameter independent cbYBE: 
\[
\frac{1}{u-v}\left(C_{12}^{(11)}\kappa_{1}\kappa_{2}-\kappa_{1}\kappa_{2}C_{12}^{(22)}\right)+\frac{1}{u+v}\left(\kappa_{1}C_{12}^{(21)}\kappa_{2}-\kappa_{2}C_{12}^{(12)}\kappa_{1}\right)=0.
\]
We have seen at the proof of Proposition \ref{lem:Adk} and Corollary
\ref{cor:inv} that if $\kappa$ satisfies the equation above then
it defines an involution $\alpha$ for which $\rho^{(1)}(\alpha(X))=\mathrm{Ad}_{\kappa}(\rho^{(2)}(X))$.
Let $\mathfrak{h}$ be the invariant sub-algebra of $\alpha$. We
can fix the normalization by
\[
\mathrm{Tr}\left(\tilde{\kappa}(u)\kappa^{-1}\right)=d
\]

\begin{prop}
If \textup{$\mathfrak{h}$ is semi-simple Lie-algebra then $\tilde{\kappa}(u)=\kappa$.}
\end{prop}
\begin{proof}
Let the asymptotic expansion of $\tilde{\kappa}(u)$ be
\[
\tilde{\kappa}(u)=\kappa+\frac{1}{u^{r}}k^{(r)}+\mathcal{O}(u^{-(r+1)})
\]
Substituting this into 
\begin{eqnarray}
 & \frac{x}{u-v}\left(C_{12}^{(11)}\tilde{\kappa}_{1}(u/x)\tilde{\kappa}_{2}(v/x)-\tilde{\kappa}_{1}(u/x)\tilde{\kappa}_{2}(v/x)C_{12}^{(22)}\right)+\label{eq:cbYBE-1}\\
+ & \frac{x}{u+v}\left(\tilde{\kappa}_{1}(u/x)C_{12}^{(21)}\tilde{\kappa}_{2}(v/x)-\tilde{\kappa}_{2}(v/x)C_{12}^{(12)}\tilde{\kappa}_{1}(u/x)\right) & =0\nonumber 
\end{eqnarray}
the order $x^{r+1}$ term is
\begin{eqnarray}
 & \frac{x^{r+1}}{u^{r}(u-v)} & \left(C^{(11)}k_{1}^{(r)}\kappa_{2}-k_{1}^{(r)}\kappa_{2}C^{(22)}\right)+\nonumber \\
+ & \frac{x^{r+1}}{v^{r}(u-v)} & \left(C^{(11)}\kappa_{1}k_{2}^{(r)}-\kappa_{1}k_{2}^{(r)}C^{(22)}\right)+\nonumber \\
+ & \frac{x^{r+1}}{u^{r}(u+v)} & \left(k_{1}^{(r)}C^{(21)}\kappa_{2}-\kappa_{2}C^{(12)}k_{1}^{(r)}\right)+\nonumber \\
+ & \frac{x^{r+1}}{v^{r}(u+v)} & \left(\kappa_{1}C^{(21)}k_{2}^{(r)}-k_{2}^{(r)}C^{(12)}\kappa_{1}\right)=0\label{eq:sys1}
\end{eqnarray}
Using the same argument as in the proof of Theorem \ref{lem:univ},
one can prove that $k^{(r)}=0$ which implies that $\tilde{\kappa}(u)=\kappa$.
\end{proof}
\begin{prop}
If \textup{$\mathfrak{h}$ is reductive Lie-algebra then the }asymptotic
expansion of \textup{$\tilde{\kappa}(u)$ has to be
\begin{equation}
\tilde{\kappa}(u)=\kappa+\frac{a_{0}}{u}\rho(X_{0})\kappa+\mathcal{O}(u^{-2})\label{eq:expansonkappa}
\end{equation}
where $X_{0}$ is a central element of $\mathfrak{h}$ for which $\left\langle \rho(X_{0}),\rho(X_{0})\right\rangle =1$
and $a_{0}\in\mathbb{C}$. If we fix $a_{0}$ then $\tilde{\kappa}(u$)
is unique.}
\end{prop}
\begin{proof}
Putting
\[
\tilde{\kappa}(u)=\kappa+\frac{1}{u}k^{(1)}+\mathcal{O}(u^{-(2)})
\]
into equation \eqref{eq:cbYBE-1} we obtain the equation \eqref{eq:sys1}
for $r=1$. Using the same argument as the proof of Theorem \ref{lem:univ},
one can prove that $k^{(1)}\kappa^{-1}$ has to be a central element
of $\mathfrak{h}$ which means that $k^{(1)}=a_{0}\rho(X_{0})\kappa$.

The proof of the universality is the same as for $K(u)$.
\end{proof}
\begin{cor}
Let \textup{$\tilde{\kappa}(u,a_{0}$) be the $\kappa$-matrix with
a fixed parameter $a_{0}$ and $\tilde{\kappa}(u)=\tilde{\kappa}(u,1)$
then $\tilde{\kappa}(u,a_{0})=\tilde{\kappa}(u/a_{0},1)=\tilde{\kappa}(u/a_{0})$.}
\end{cor}
We can connect the parameter $a_{0}$ of \eqref{eq:expansonkappa}
to the function $a(\hbar)$. From equations \eqref{eq:k1},\eqref{eq:k2}
and \eqref{eq:expansonkappa} we can obtain that
\[
a(\hbar)=\frac{1}{\hbar}\left(a_{0}+\mathcal{O}(\hbar)\right).
\]

In summary, if the residual symmetry algebra is semi-simple then the
classical $\kappa$-matrix is fixed by the symmetry, and the quantum
correction to K-matrix is fixed by this $\kappa$ and the bYBE (up
to a normalization)
\[
\kappa\Longrightarrow K(u,\hbar).
\]
However, if the residual symmetry algebra is not semi-simple but reductive
then the classical $\kappa$-matrix is not totally fixed by the symmetry
since it has a free parameter $a_{0}$, and the K-matrix has a dynamical
parameter $a(\hbar)$ which is not fixed by the $\kappa(a_{0}u)$
and the bYBE
\[
\left\{ \kappa(a_{0}u),a(\hbar)\right\} \Longrightarrow K(u,\hbar,a(\hbar)).
\]

\subsection{Not quasi classical K-matrices}

There also exist non quasi-classical solutions of the bYBE. If $\mathfrak{g}=\mathfrak{sl}(n)$
and $\rho^{(1)}=\rho^{(2)}=\rho$, where $\rho$ is the defining representation
then 
\[
K(u)=\kappa+\frac{1}{u}1
\]
is a solution for $\kappa^{2}=0$ i.e. $K(u)$ is not quasi classical.
This K-matrix is satisfies the unitary condition:
\[
K(u)K(-u)=-\frac{1}{u^{2}}1.
\]
The residual symmetry algebra is not reductive but a semi-direct sum
of a solvable and a reductive Lie-algebras
\[
\mathfrak{h}=\mathfrak{h}_{s}\oplus\mathfrak{h}_{r}
\]
where $\mathfrak{h}_{s}$ is solvable and $\mathfrak{h}_{r}$ is reductive.
Since 
\[
\left[\mathfrak{h}_{s},\mathfrak{h}_{s}\right]\subset\mathfrak{h}_{s},\qquad\left[\mathfrak{h}_{s},\mathfrak{h}_{r}\right]\subset\mathfrak{h}_{s},\qquad\left[\mathfrak{h}_{r},\mathfrak{h}_{r}\right]=\mathfrak{h}_{r},
\]
the above sum is semi-direct. With a suitable base change, $\kappa$
can always be brought into Jordan canonical form:
\[
\kappa=\left(\begin{array}{cc|c|cc|cc|c|cc}
0 & 1 &  &  &  &  &  & \\
0 & 0 &  &  &  &  &  & \\
\hline  &  & \ddots &  &  &  &  & \\
\hline  &  &  & 0 & 1 &  &  & \\
 &  &  & 0 & 0 &  &  & \\
\hline  &  &  &  &  & 0 & 0 & \\
 &  &  &  &  & 0 & 0 & \\
\hline  &  &  &  &  &  &  & \ddots\\
\hline  &  &  &  &  &  &  &  & 0 & 0\\
 &  &  &  &  &  &  &  & 0 & 0
\end{array}\right)
\]
where there are $k$ non-trivial block. The reductive sub-algebra
is $\mathfrak{h}_{r}=\mathfrak{sl}(n-2k)$. The solvable part can
be written in a direct sum of subspaces: $\mathfrak{h}_{s}=\mathfrak{h}_{2}+\mathfrak{h}_{D}+\mathfrak{h}_{+}+\mathfrak{h}_{-}$
where 
\begin{eqnarray*}
\mathfrak{h}_{2} & = & \left\{ E_{2a-1,2b}|a=1,\dots,k;b=1,\dots,k\right\} \\
\mathfrak{h}_{D} & = & \left\{ H_{a}=E_{2a-1,2a-1}+E_{2a,2a}-\frac{2}{n-2k}\sum_{i=2k+1}^{n}E_{i,i}|a=1,\dots,k\right\} \\
\mathfrak{h}_{+} & = & \left\{ E_{2a-1,i}|a=1,\dots,k;i=2k+1,\dots,n\right\} \\
\mathfrak{h}_{-} & = & \left\{ E_{i,2a}|a=1,\dots,k;i=2k+1,\dots,n\right\} 
\end{eqnarray*}
where $E_{i,j}$s are the elementary matrices: $\left(E_{i,j}\right)_{ab}=\delta_{ia}\delta_{jb}$.
The Lie-bracket of these are the following:
\begin{eqnarray*}
\left[\mathfrak{h}_{2},\mathfrak{h}_{2}\right] & = & \left[\mathfrak{h}_{2},\mathfrak{h}_{+}\right]=\left[\mathfrak{h}_{2},\mathfrak{h}_{-}\right]=\left[\mathfrak{h}_{D},\mathfrak{h}_{D}\right]=\left[\mathfrak{h}_{+},\mathfrak{h}_{+}\right]=\left[\mathfrak{h}_{-},\mathfrak{h}_{-}\right]=0\\
\left[\mathfrak{h}_{D},\mathfrak{h}_{2}\right] & \subseteq & \mathfrak{h}_{2}\\
\left[\mathfrak{h}_{D},\mathfrak{h}_{+}\right] & \subseteq & \mathfrak{h}_{+}\\
\left[\mathfrak{h}_{D},\mathfrak{h}_{-}\right] & \subseteq & \mathfrak{h}_{-}\\
\left[\mathfrak{h}_{+},\mathfrak{h}_{-}\right] & \subseteq & \mathfrak{h}_{2}
\end{eqnarray*}
Therefore $\left[\mathfrak{h}_{s},\mathfrak{h}_{s}\right]:=\mathfrak{h}_{1}=\mathfrak{h}_{2}+\mathfrak{h}_{+}+\mathfrak{h}_{-}$
and $\left[\mathfrak{h}_{1},\mathfrak{h}_{1}\right]=\mathfrak{h}_{2}$
which is a commutative Lie-algebra.

\section{K-matrices with general boundary space}

In this section we investigate K-matrices with $d_{B}>1$.
\begin{defn}
Let $K(u)=E_{ij}\otimes\Psi^{ij}(u)\in\mathrm{End}(\mathbb{C}^{d}\otimes\mathbb{C}^{d_{B}})$
where $E_{ij}$ are the elementary matrices of $\mathrm{End}(\mathbb{C}^{d})$.
The K-matrix $K(u)$ is irreducible if there is no proper invariant
subspace $\mathbb{C}^{d_{0}}\subset\mathbb{C}^{d_{B}}$ of $\Psi^{ij}(u)$
for all $i,j$.
\end{defn}
\begin{cor}
If $K(u)\in\mathrm{End}(\mathbb{C}^{d}\otimes\mathbb{C}^{d_{B}})$
is an irreducible K-matrix then $\tilde{\kappa}=\kappa\otimes1$ where
$\kappa\in\mathrm{End}(\mathbb{C}^{d})$.
\end{cor}
\begin{proof}
The bYBE in the $v\to\infty$ limit reads as
\[
\left[K(u)_{13},\tilde{\kappa}_{23}\right]=0.
\]
We can write the K-matrix and the matrix $\tilde{\kappa}$ as
\begin{eqnarray*}
K(u) & = & E_{ij}\otimes\Psi^{ij}(u)\\
\tilde{\kappa} & = & E_{ij}\otimes\psi^{ij}
\end{eqnarray*}
where $E_{ij}$s are the elementary matrices of $\mathrm{End}(\mathbb{C}^{d})$
and $\Psi^{ij}(u),\psi^{ij}\in\mathrm{End}(\mathbb{C}^{d_{B}})$ for
all $i,j$. Putting this into the equation above, we obtain that
\[
\left[\Psi^{ij}(u),\psi^{kl}\right]=0
\]
for all $i,j,k,l$. Since $K(u)$ is irreducible, the matrices $\psi^{kl}$
have to be proportional to the identity therefore $\tilde{\kappa}=\kappa\otimes1$.
\end{proof}
\begin{prop}
Let $K(u)\in\mathrm{End}(\mathbb{C}^{d}\otimes\mathbb{C}^{d_{B}})$
be an irreducible K-matrix in the representation $\left(\rho^{(1)},\rho^{(2)}\right)$
then there exists a Lie-algebra involution $\alpha:\mathfrak{g}\to\mathfrak{g}$,
$\alpha^{2}=\mathrm{id}_{\mathfrak{g}}$ for which \textup{$\mathrm{Ad}_{\kappa}(\rho^{(2)}(X))=\rho^{(1)}(\alpha(X))$}.
\end{prop}
\begin{proof}
From the previous corollary, we know that the leading order of an
irreducible K-matrix is $\kappa\otimes1$ therefore the leading non-trivial
term of the bYBE
\begin{eqnarray}
R_{12}^{(11)}\left(\frac{u-v}{x}\right)K_{13}\left(\frac{u}{x}\right)R_{21}^{(12)}\left(\frac{u+v}{x}\right)K_{23}\left(\frac{v}{x}\right) & =\label{eq:bYBEA}\\
K_{23}\left(\frac{v}{x}\right)R_{12}^{(12)}\left(\frac{u+v}{x}\right)K_{13}\left(\frac{u}{x}\right)R_{21}^{(22)}\left(\frac{u-v}{x}\right)
\end{eqnarray}
is the same as \eqref{eq:classlimit}. Therefore the proof of this
proposition is the same as it was for the scalar boundary.
\end{proof}
\begin{lem}
Let $K(u)\in\mathrm{End}(\mathbb{C}^{d}\otimes\mathbb{C}^{d_{B}})$
be an irreducible K-matrix in the representation $\left(\rho^{(1)},\rho^{(2)}\right)$
and $\alpha$ is the Lie-algebra involution for which \textup{$\mathrm{Ad}_{\kappa}(\rho^{(2)}(X))=\rho^{(1)}(\alpha(X))$}
then the asymptotic expansion of $K(u)$ has to be
\[
K(u)=\kappa\otimes1+\frac{1}{u}\left(\frac{1}{2}c^{(\mathfrak{h},1)}\otimes1+D\otimes1+2C^{(\mathfrak{h},1B)}\right)\kappa\otimes1+\mathcal{O}(u^{-2}),
\]
where $\rho^{(B)}:\mathfrak{h}\to\mathrm{End}(\mathbb{C}^{d_{B}})$
is a representation of $\mathfrak{h}$ which is the invariant subalgebra
of $\alpha$ and $D\in\mathrm{End}(\mathbb{C}^{d})$ for which $\left[D,\rho^{(1)}\left(\mathfrak{g}\right)\right]=0$.
\end{lem}
\begin{proof}
Let the asymptotic expansions be
\begin{eqnarray*}
K(u) & = & \kappa\otimes1+\frac{1}{u}k^{(1)}+\mathcal{O}(u^{-2})\\
R^{(ij)}(u) & = & 1\otimes1+\frac{1}{u}C^{(ij)}+\frac{1}{u^{2}}D^{(ij)}+\mathcal{O}(u^{-2})
\end{eqnarray*}
We can substitute this to \eqref{eq:bYBEA} and the next to the leading
order in $x$ is the following:
\begin{eqnarray}
\frac{(1)}{\left(u-v\right)^{2}}+\frac{(2)}{\left(u+v\right)^{2}}+\frac{(3)}{u^{2}-v^{2}}+\nonumber \\
+\frac{(4)}{u\left(u-v\right)}+\frac{(5)}{v\left(u-v\right)}+\frac{(6)}{u\left(u+v\right)}+\frac{(7)}{v\left(u+v\right)}+\frac{(8)}{uv} & = & 0,\label{eq:x2}
\end{eqnarray}
where 
\begin{eqnarray*}
(1) & = & D_{12}^{(11)}\kappa_{1}\kappa_{2}-\kappa_{1}\kappa_{2}D_{12}^{(22)}\\
(2) & = & \kappa_{1}D_{12}^{(21)}\kappa_{2}-\kappa_{2}D_{12}^{(12)}\kappa_{1}\\
(3) & = & C_{12}^{(11)}\kappa_{1}C_{12}^{(21)}\kappa_{2}-\kappa_{2}C_{12}^{(12)}\kappa_{1}C_{12}^{(22)}\\
(4) & = & C_{12}^{(11)}k_{13}^{(1)}\kappa_{2}-k_{13}^{(1)}\kappa_{2}C_{12}^{(22)}\\
(5) & = & C_{12}^{(11)}\kappa_{1}k_{23}^{(1)}-\kappa_{1}k_{23}^{(1)}C_{12}^{(22)}\\
(6) & = & k_{13}^{(1)}C_{12}^{(21)}\kappa_{2}-\kappa_{2}C_{12}^{(12)}k_{13}^{(1)}\\
(7) & = & \kappa_{1}C_{12}^{(21)}k_{23}^{(1)}-k_{23}^{(1)}C_{12}^{(12)}\kappa_{1}\\
(8) & = & \left[k_{13}^{(1)},k_{23}^{(1)}\right]
\end{eqnarray*}
Multipling equation \eqref{eq:x2} by $uv(u-v)^{2}(u+v)^{2}$, we
obtain that
\begin{eqnarray}
u^{4}\left((5)+(7)+(8)\right)+v^{4}\left(-(4)+(6)+(8)\right)+\nonumber \\
+u^{3}v\left((1)+(2)+(3)+(4)+(5)+(6)-(7)\right)+\nonumber \\
+uv^{3}\left((1)+(2)-(3)-(4)-(5)-(6)+(7)\right)+\nonumber \\
+u^{2}v^{2}\left(2(1)-2(2)+(4)-(5)-(6)-(7)-2(8)\right)=0\label{eq:r2}
\end{eqnarray}
Let us start with the $u^{4}$ term:
\[
C_{12}^{(11)}\kappa_{1}k_{23}^{(1)}-\kappa_{1}k_{23}^{(1)}C_{12}^{(22)}+\kappa_{1}C_{12}^{(21)}k_{23}^{(1)}-k_{23}^{(1)}C_{12}^{(12)}\kappa_{1}+\left[k_{13}^{(1)},k_{23}^{(1)}\right]=0
\]
Multipling this by $\kappa_{1}^{-1}\kappa_{2}^{-1}$ from the right.
\[
\left[C_{12}^{(11)}+\kappa_{1}C_{12}^{(21)}\kappa_{1}^{-1},k_{23}^{(1)}\kappa_{2}^{-1}\right]+\left[k_{13}^{(1)}\kappa_{1}^{-1},k_{23}^{(1)}\kappa_{2}^{-1}\right]=0
\]
Since
\[
C_{12}^{(11)}+\kappa_{1}C_{12}^{(21)}\kappa_{1}^{-1}=2C_{12}^{(\mathfrak{h},11)}
\]
we obtain that
\begin{equation}
2\left[C_{12}^{(\mathfrak{h},11)},k_{23}^{(1)}\kappa_{2}^{-1}\right]+\left[k_{13}^{(1)}\kappa_{1}^{-1},k_{23}^{(1)}\kappa_{2}^{-1}\right]=0\label{eq:x2a}
\end{equation}
Let us use the following notation:
\[
\tilde{k}^{(1)}=k^{(1)}\left(\kappa^{-1}\otimes1\right)=Y_{a}^{(1)}\otimes Z^{a}+\bar{Y}_{\bar{a}}^{(1)}\otimes\bar{Z}^{\bar{a}}
\]
where $Z^{a},\bar{Z}^{\bar{a}}\in\mathrm{End}(\mathbb{C}^{d_{B}})$.
Substituting this into \eqref{eq:x2a} we obtain that
\begin{eqnarray}
2B^{ab}Y_{a}^{(1)}\otimes\left[Y_{b}^{(1)},Y_{c}^{(1)}\right]\otimes Z^{c}+2B^{ab}Y_{a}^{(1)}\otimes\left[Y_{b}^{(1)},\bar{Y}_{\bar{c}}^{(1)}\right]\otimes\bar{Z}^{\bar{c}}+\nonumber \\
+Y_{a}^{(1)}\otimes Y_{b}^{(1)}\otimes\left[Z^{a},Z^{b}\right]+Y_{a}^{(1)}\otimes\bar{Y}_{\bar{b}}^{(1)}\otimes\left[Z^{a},\bar{Z}^{\bar{b}}\right]+\nonumber \\
+\bar{Y}_{\bar{a}}^{(1)}\otimes Y_{b}^{(1)}\otimes\left[\bar{Z}^{\bar{a}},Z^{b}\right]+\bar{Y}_{\bar{a}}^{(1)}\otimes\bar{Y}_{\bar{b}}^{(1)}\otimes\left[\bar{Z}^{\bar{a}},\bar{Z}^{\bar{b}}\right]=0\label{eq:eq1}
\end{eqnarray}
If we apply $\left\langle Y_{d}^{(1)},\cdot\right\rangle _{1}\otimes\left\langle Y_{e}^{(1)},\cdot\right\rangle _{1}\otimes\mathrm{id}$
then we get
\[
2B_{fe}f_{dc}^{f}Z^{c}+B_{ad}B_{be}\left[Z^{a},Z^{b}\right]=0
\]
where we used that $C_{ab}^{(1)}=c^{(1)}B_{ab}$. Using the definition
\[
Z_{a}=\frac{1}{2}B_{ab}Z^{b}
\]
we can obtain that
\[
\left[Z_{d},Z_{e}\right]=f_{de}^{c}Z_{c}
\]
which implies that there exists a representation ($\rho^{(B)}$) of
$\mathfrak{h}$ for which $\rho^{(B)}(X_{a})=Z_{a}$.

Appling $\left\langle \bar{Y}_{\bar{d}}^{(1)},\cdot\right\rangle _{1}\otimes\left\langle Y_{e}^{(1)},\cdot\right\rangle _{1}\otimes\mathrm{id}$
to \eqref{eq:eq1} we obtain that 
\begin{equation}
\left[\bar{Z}^{\bar{d}},Z^{e}\right]=0\label{eq:com1}
\end{equation}
and using $\left\langle \bar{Y}_{\bar{d}}^{(1)},\cdot\right\rangle _{1}\otimes\left\langle \bar{Y}_{\bar{e}}^{(1)},\cdot\right\rangle _{1}\otimes\mathrm{id}$,
we get
\begin{equation}
\left[\bar{Z}^{\bar{d}},\bar{Z}^{\bar{e}}\right]=0.\label{eq:com2}
\end{equation}
Using these, $\bar{Y}_{\bar{a}}^{(1)}\otimes\bar{Z}^{\bar{a}}$ can
be written as $\bar{Y}_{\bar{a}}^{(1)}\otimes\bar{Z}^{\bar{a}}=y^{i}\otimes U_{i}$
where $\left[U_{i},\rho^{(B)}(X)\right]=0$, $\left[U_{i},U_{j}\right]=0$,
$\left\langle U_{i},U_{j}\right\rangle _{B}=\delta_{ij}$ for all
$X\in\mathfrak{h}$ and $y^{i}\in\rho^{(1)}(\mathfrak{f})\oplus\bar{V}^{(1)}$.
Using $\left\langle Y_{d}^{(1)},\cdot\right\rangle \otimes\mathrm{id}\otimes\left\langle U_{i},\cdot\right\rangle $
on \eqref{eq:eq1}, we obtain that
\begin{equation}
\left[Y_{d}^{(1)},y^{i}\right]=0.\label{eq:yd}
\end{equation}

From the $u^{3}v$ and $uv^{3}$ terms of \eqref{eq:r2} we can obtain
that

\[
(3)+(4)+(5)+(6)-(7)=0
\]
Multiplying this with $\kappa_{1}^{-1}\kappa_{2}^{-1}$ from the right,
we get
\[
\left[C_{12}^{(\mathfrak{f},11)},C_{12}^{(\mathfrak{h},11)}\right]+\left[C_{12}^{(\mathfrak{f},11)},\tilde{k}_{13}^{(1)}+\tilde{k}_{23}^{(1)}\right]=0
\]
The first term can be written as
\[
\left[C_{12}^{(\mathfrak{f},11)},C_{12}^{(\mathfrak{h},11)}\right]=-\frac{1}{2}\left[C_{12}^{(\mathfrak{f},11)},c^{(\mathfrak{h},1)}\otimes1+1\otimes c^{(\mathfrak{h},1)}\right]
\]
Using the definition $\bar{Y}_{\bar{a}}^{(1)}\otimes\bar{Z}^{\bar{a}}=y^{i}\otimes U_{i}=\tilde{y}^{i}\otimes U_{i}+\frac{1}{2}c^{(\mathfrak{h},1)}\otimes1$
the $\tilde{k}^{(1)}$ can be written as
\[
\tilde{k}^{(1)}=\tilde{y}^{i}\otimes U_{i}+\frac{1}{2}c^{(\mathfrak{h},1)}\otimes1+2B^{ab}\rho^{(1)}\left(X_{a}\right)\otimes\rho^{(B)}\left(X_{b}\right).
\]
Substituting this into the equation above, we obtain that
\[
B^{\alpha\beta}\left[Y_{\alpha}^{(1)},\tilde{y}^{i}\right]\otimes Y_{\beta}^{(1)}\otimes U_{i}+B^{\alpha\beta}Y_{\alpha}^{(1)}\otimes\left[Y_{\beta}^{(1)},\tilde{y}^{i}\right]\otimes U_{i}=0
\]
Using $\left\langle Y_{\gamma}^{(1)},\cdot\right\rangle \otimes\mathrm{id}\otimes\left\langle U_{i},\cdot\right\rangle $,
we get
\begin{equation}
\left[Y_{\beta}^{(1)},\tilde{y}^{i}\right]=0.\label{eq:yb}
\end{equation}
From \eqref{eq:yd} and \eqref{eq:yb} we can see that $\left[\rho^{(1)}(X),\tilde{y}^{i}\right]=0$
for all $X\in\mathfrak{g}$ therefore $\tilde{y}^{i}=d^{ij}D_{j}$
where $\left[D_{j},\rho^{(1)}(X)\right]=0$ for all $X\in\mathfrak{g}$,
$\left\langle D_{i},D_{j}\right\rangle _{1}=\delta_{ij}$ and $D_{i}\in\rho^{(1)}(\mathfrak{f})\oplus\bar{V}^{(1)}$. 

In summary,
\begin{eqnarray}
\tilde{k}^{(1)}=D_{i}\otimes U^{i}+\frac{1}{2}c^{(\mathfrak{h},1)}\otimes1+2B^{ab}\rho^{(1)}\left(X_{a}\right)\otimes\rho^{(B)}\left(X_{b}\right)=\label{eq:kt}\\
=D_{i}\otimes U^{i}+\frac{1}{2}c^{(\mathfrak{h},1)}\otimes1+2C^{(\mathfrak{h},1B)}
\end{eqnarray}
where $U^{i}=d^{ji}U_{j}$.

In the following we prove that $U^{i}$ has to be proportional to
$1$. For this, we use the higher order terms of \eqref{eq:bYBEA}.
Let the asymptotic expansions be
\begin{eqnarray*}
K(u) & = & \sum_{r=0}^{\infty}\frac{1}{u^{r}}k^{(r)},\\
R^{(ij)}(u) & = & \sum_{r=0}^{\infty}\frac{1}{u^{r}}C^{(r)(ij)},
\end{eqnarray*}
where $C^{(0)(ij)}=1\otimes1$, $C^{(1)(ij)}=C^{(ij)}$, $k^{(0)}=\kappa\otimes1$.The
$x^{r+1}$ order terms of \eqref{eq:bYBEA} read as
\[
\frac{1}{u^{a}v^{b}(u-v)^{c}(u+v)^{d}}\left(C_{12}^{(c)(11)}k_{13}^{(a)}C_{12}^{(d)(21)}k_{23}^{(b)}-k_{23}^{(b)}C_{12}^{(d)(12)}k_{13}^{(a)}C_{12}^{(c)(22)}\right)
\]
where $a+b+c+d=r+1$. Because of $k^{(0)}=\kappa\otimes1$, the $a=r+1$
and $b=r+1$ terms are trivial therefore $a,b\leqq r$. Multipling
these by $u^{r}v^{r}(u-v)^{r+1}(u+v)^{r+1}$, the spectral parameter
dependencies are
\begin{eqnarray*}
u^{r-a}v^{r-b}(u-v)^{r+1-c}(u+v)^{r+1-d}\Bigl( & C_{12}^{(c)(11)}k_{13}^{(a)}C_{12}^{(d)(21)}k_{23}^{(b)}-\\
 & -k_{23}^{(b)}C_{12}^{(d)(12)}k_{13}^{(a)}C_{12}^{(c)(22)} & \Bigr)
\end{eqnarray*}
Expanding the brackets, the highest order term in $u$ is $u^{2r+1+b}v^{r-b}$.
We concentrate on the $u^{3r+1}$ terms i.e. when $b=r$. From this
we can obtain the following equation:

\[
\left(C_{12}^{(11)}\kappa_{1}k_{23}^{(r)}-k_{23}^{(r)}\kappa_{1}C_{12}^{(22)}\right)+\left(\kappa_{1}C_{12}^{(21)}k_{23}^{(r)}-k_{23}^{(r)}C_{12}^{(12)}\kappa_{1}\right)+\left[k_{13}^{(1)},k_{23}^{(r)}\right]=0
\]
Multiply this by $\kappa_{1}^{-1}\kappa_{2}^{-1}$ from the right,
we get
\[
2\left[C_{12}^{(\mathfrak{h},11)},\tilde{k}_{23}^{(r)}\right]+\left[\tilde{k}_{13}^{(1)},\tilde{k}_{23}^{(r)}\right]=0
\]
Using \eqref{eq:kt} and applying $\left\langle D_{i},\cdot\right\rangle \otimes\mathrm{id}\otimes\mathrm{id}$,
we obtain that
\begin{equation}
\left[1\otimes U^{i},\tilde{k}^{(r)}\right]=0\label{eq:temp}
\end{equation}
Since this is true for all $r$, the equation \eqref{eq:temp} implies
that
\[
\left[1\otimes U^{i},K(u)\right]=0
\]
If $K(u)$ is irreducible then $U^{i}=e^{i}1$ therefore
\begin{equation}
\tilde{k}^{(1)}=D\otimes1+\frac{1}{2}c^{(\mathfrak{h},1)}\otimes1+2C^{(\mathfrak{h},1B)}\label{eq:ktv}
\end{equation}
where $D=e^{i}D_{i}$ and we wanted to prove this.
\end{proof}
\begin{thm}
Let $K(u)$ be a quasi classical, irreducible $(\mathfrak{g},\mathfrak{h})$
symmetric K-matrix in the representation $\left(\rho^{(1)},\rho^{(2)}\right)$
then $\left(\mathfrak{g},\mathfrak{h}\right)$ is a symmetric pair.
\end{thm}
\begin{proof}
Similarly to the scalar boundary case, the leading order of $K(u)$
defines a Lie-algebra involution $\alpha$, for which\textit{ }$\mathrm{Ad}_{\kappa}(\rho^{(2)}(X))=\rho^{(1)}(\alpha(X))$\textit{.}
This involution can be used for a $\mathbb{Z}_{2}$ graded decomposition:
$\mathfrak{g}=\mathfrak{h}_{0}\oplus\mathfrak{f}$ where $\alpha(X^{(+)})=+X^{(+)}$
and $\alpha(X^{(-)})=-X^{(-)}$ for all $X^{(+)}\in\mathfrak{h}_{0}$
and $X^{(-)}\in\mathfrak{f}$ therefore $\left(\mathfrak{g},\mathfrak{h}_{0}\right)$
is a symmetric pair.

Let us assume that $X\in\mathfrak{h}$ which implies that there exists
a representation $\tilde{\rho}$ of $\mathfrak{h}$ such that 
\[
K(u)\left(\rho^{(2)}(X)\otimes1\right)-\left(\rho^{(1)}(X)\otimes1\right)K(u)+\left[K(u),1\otimes\tilde{\rho}(X)\right]=0.
\]
Going to the $u\to\infty$ limit, we obtain that $\kappa\rho^{(2)}(X)=\rho^{(1)}(X)\kappa$
i.e. $\mathrm{Ad}_{\kappa}(\rho^{(2)}(X))=\rho^{(1)}(X)$ which means
$\alpha(X)=X$ i.e $X\in\mathfrak{h}_{0}$ therefore $\mathfrak{h}\subseteq\mathfrak{h}_{0}$.

Now, we take the $v\to\infty$ limit of the bYBE:
\begin{eqnarray*}
\frac{1}{v}\left(K_{13}(u)C_{12}^{(21)}\kappa_{2}-C_{12}^{(11)}K_{13}(u)\kappa_{2}+K_{13}(u)k_{23}^{(1)}\right)+\mathcal{O}(v^{-2}) & =\\
\frac{1}{v}\left(\kappa_{2}C_{12}^{(12)}K_{13}(u)-\kappa_{2}K_{13}(u)C_{12}^{(22)}+k_{23}^{(1)}K_{13}(u)\right)+\mathcal{O}(v^{-2})
\end{eqnarray*}
Multipling this by $\kappa_{2}^{-1}$ from the right and using \eqref{eq:ktv},
we obtain that
\[
K_{13}(u)C_{12}^{(\mathfrak{h}_{0},21)}-C_{12}^{(\mathfrak{h}_{0},11)}K_{13}(u)+\left[K_{13}(u),C_{23}^{(\mathfrak{h}_{0},1B)}\right]=0
\]
which is equivalent to
\[
K(u)\left(\rho^{(2)}(X)\otimes1\right)-\left(\rho^{(1)}(X)\otimes1\right)K(u)+\left[K(u),1\otimes\rho^{(B)}(X)\right]=0
\]
for all $X\in\mathfrak{h}_{0}$ which implies that $\mathfrak{h}_{0}\subseteq\mathfrak{h}$.
We have seen previously that $\mathfrak{h}\subseteq\mathfrak{h}_{0}$
therefore $\mathfrak{h}_{0}=\mathfrak{h}$ i.e. $\left(\mathfrak{g},\mathfrak{h}\right)$
is a symmetric pair.
\end{proof}

\section{Conclusion}

In this paper we derived directly from the boundary Yang-Baxter equation
that the possible residual symmetry algebras of the quasi-classical
K-matrices have to be invariant sub-algebras of Lie-algebra involutions.
It was also proved that if the boundary vector space is one dimensional
then these K-matrices are universal (up to a normalization) when the
residual sub-algebra is semi-simple, and for the non semi-simple ones
the K-matrices have a free parameter. In the following it might be
interesting to try to generalize these statements and proofs to the
trigonometric cases where similar classification seemed to be present
\cite{Regelskis:2016iaf,Nepomechie:2018dsn}.

In addition, a not quasi-classical K-matrix was briefly examined.
We have seen that its residual symmetry algebra is a semi-direct sum
of a reductive and a solvable Lie-algebra. For further work, the classification
of these may also be interesting. It is possible that they also play
a role in the classification of the integrable initial states of spin
chains \cite{Piroli:2018ksf,Piroli:2018don,Pozsgay:2018dzs}.

\section*{Acknowledgment}

I thank Zoltán Bajnok for the useful discussions and for reading the
manuscript. The work was supported by the NKFIH 116505 Grant.

\section*{References}

\bibliographystyle{alphaurl}
\bibliography{Kmatrix}

\end{document}